\documentclass{IEEEtran}
\usepackage[dvips]{graphicx}
\usepackage[USenglish]{babel}
\usepackage{amsmath,amssymb,amsthm,dsfont,url,braket,ctable,hyperref,commath}
\usepackage[percent]{overpic}
\usepackage{mathtools}

\newcommand{\Tr}{\operatorname{Tr}}
\newcommand{\openone}{\mathds{1}}

\newtheorem{thm}{Theorem}

\newtheorem{lmm}{Lemma}

\DeclarePairedDelimiter{\floor}{\lfloor}{\rfloor}

\begin{document}

\title{Tradeoff relations between accessible information, informational power, and purity}

\author{Michele        Dall'Arno        and        Francesco
  Buscemi\thanks{M. Dall'Arno is with the Centre for Quantum
    Technologies,  National   University  of   Singapore,  3
    Science   Drive    2,   117543,    Singapore   (e--mail:
    cqtmda@nus.edu.sg).}\thanks{F.   Buscemi   is  with  the
    Graduate  School  of   Informatics,  Nagoya  University,
    Chikusa-ku,    464-8601,    Nagoya,   Japan    (e--mail:
    buscemi@i.nagoya-u.ac.jp).}\thanks{Copyright   (c)  2017
    IEEE.  Personal  use  of  this  material  is  permitted.
    However, permission  to use this material  for any other
    purposes must  be obtained  from the  IEEE by  sending a
    request to pubs-permissions@ieee.org.}}

\maketitle

\begin{abstract}
  The  accessible information  and  the informational  power
  quantify  the maximum  amount of  information that  can be
  extracted  from  a  quantum  ensemble  and  by  a  quantum
  measurement,  respectively.   Here,   we  investigate  the
  tradeoff between the accessible information (informational
  power, respectively) and  the purity of the  states of the
  ensemble (the elements  of the measurement, respectively).
  Under any  given {\em  lower} bound on  the purity,  i) we
  compute the  minimum informational power and  show that it
  is  attained  by   the  depolarized  uniformly-distributed
  measurement; ii) we  give a lower bound  on the accessible
  information.   Under any  given {\em  upper} bound  on the
  purity, i)  we compute the maximum  accessible information
  and show  that it is  attained by an ensemble  of pairwise
  commuting  states  with  at  most  two  distinct  non-null
  eigenvalues;  ii) we  give a  lower bound  on the  maximum
  informational  power. The  present results  provide, as  a
  corollary, novel  sufficient conditions for  the tightness
  of the  Jozsa-Robb-Wootters lower bound to  the accessible
  information.
\end{abstract}


\IEEEPARstart{W}{e}  address  the problem  of  communicating
classical  information   over  the  most   general  physical
channel, that is a quantum channel (classical channels being
a particular instance of  the quantum case).  In particular,
we  consider the  case in  which  the sender  is allowed  to
encode a classical random variable  $X$ on a quantum system,
which is then  transmitted to a receiver  and measured, thus
producing  an output  classical  random  variable $Y$.   The
encoding here  produces an  ensemble of quantum  states, one
for each  letter in the input  alphabet $\mathcal{X}=\{x\}$,
whereas  the  measurement returns  a  letter  in the  output
alphabet $\mathcal{Y}=\{y\}$.

When the input ensemble is  fixed, the final measurement can
be  optimized  to  achieve  the  maximum  amount  of  mutual
input-output information $I(X;Y)$.  This quantity is defined
as    the    {\em    accessible    information}    of    the
ensemble~\cite{Wil13}.   By  direct   analogy,  the  maximum
amount of  input-output information that can  be established
for  a fixed  measurement, by  optimizing over  all possible
input ensembles, is defined as the {\em informational power}
of                                                       the
measurement~\cite{DDS11,Hol12,SS14,DBO14,Szy14,Dal14,Dal15,SS16,BDS16}. A
duality  relation  between these  two  information-theoretic
measures was established  in Ref.~\cite{DDS11}.  Within this
context,  one  is  generally   interested  in  bounding  the
accessible information and the  informational power that can
be  achieved given  some  resources, for  example for  fixed
Hilbert space dimension.

A  family  of quantum  states  or  measurement operators  is
called  ``pure''  if all  its  elements  are represented  by
rank-one    operators.     Mathematically   speaking,    the
\textit{purity} of a positive  semi-definite operator $X$ is
given  by   $P(X)=\Tr[X^2]/(\Tr[X])^2$.   Intuitively,  this
number  is  usually  considered  as a  good  proxy  for  the
``classical  uncertainty''  contained in  a  state  or in  a
measurement: the  higher the purity, the  less the classical
uncertainty. As it can be  readily shown, the purity reaches
its maximum ($P=1$) on  rank-one operators.  Our main result
is to derive analytical bounds on the accessible information
and the  informational power that  consider the purity  as a
free variable in  the problem. In this sense,  purity can be
considered as a resource, only available in limited amounts.

As  an  example,  let  us   consider  lower  bounds  on  the
accessible     information~\cite{JRW94,FC94}     and     the
informational   power~\cite{DBO14}.   These  are   typically
expressed   in    terms   of   a   quantity    called   {\em
  subentropy}~\cite{JRW94}  (see  Ref.~\cite{DDJB14}  for  a
study of its properties).  In  this sense, the subentropy of
a  given  state  $\rho$ quantifies  the  minimum  accessible
information of  any ensemble of {\em  pure} states averaging
to  $\rho$. Hence,  known  subentropy-like  lower bounds  on
informational  measures hold  only  if  the optimization  is
restricted to pure states and measurement elements.  In this
paper  we  generalize  similar  lower and  upper  bounds  by
investigating   tradeoff    relations   between   accessible
information, informational power, and  purity, which can now
be bounded by any given value $0 \le P \le 1$.

More  concretely, our  contribution is  two-fold. First,  we
consider the case in which an arbitrary {\em lower} bound on
the purity  is given.  In  this case, we derive  the minimum
informational power of any measurement when its elements are
subject  to such  a purity  constraint. We  show that  it is
attained  by the  ``depolarized Scrooge  measurement,'' that
is,   the  uniformly   depolarized,  uniformly   distributed
measurement.  We also derive a lower bound on the accessible
information.   This result  has  important connections  with
previous literature.   It proves a conjecture  formulated in
Ref.~\cite{BDS16}, where the  accessible information and the
informational  power  of   mixed  $t$-design  ensembles  and
measurements -- including  depolarized Scrooge structures --
were studied.  In the  process, our result corrects Eq.~A24,
Property~11, of Ref.~\cite{JM15}.

The  second set  of results  concerns the  case in  which an
arbitrary \textit{upper}  bound on  the purity  is enforced.
Under  this assumption,  we  derive  the maximum  accessible
information of any  ensemble when its states  are subject to
such a purity constraint. We prove  that it is attained by a
particular class  of ensembles  with commuting  states, each
with   at   most   two   different   non-null   eigenvalues.
Additionally,  we  derive  a  lower  bound  on  the  maximum
informational   power.   This   result  too   has  important
connections  with  previous  literature.  It  allows  us  to
simplify  a  proof,  given in  Ref.~\cite{HT01}  adopting  a
topological  approach,   of  the  tradeoff   between  purity
(therein  referred  to as  the  index  of coincidence  of  a
classical  probability  distribution) and  Shannon  entropy.
Moreover, our  formulation can be extended  to encompass the
case of arbitrary R\'enyi entropy, not only Shannon's.

Our  findings  have  implications  for the  problem  of  the
tightness  of the  Jozsa-Robb-Wootters  lower  bound on  the
accessible    information,     given    in     Eq.(33)    of
Ref.~\cite{JRW94}.  Prior  to this work, not  much was known
about  this  problem,  except  for the  cases  of  uniformly
distributed pure states  (Scrooge ensemble).  This contrasts
with  the  case of  the  Holevo  upper bound  on  accessible
information,  for  which  general necessary  and  sufficient
conditions for tightness  are known~\cite{Hol73, Rus02}.  As
a  consequence   of  our   results,  it  follows   that  the
Jozsa-Robb-Wootters lower bound is  also tight for uniformly
\emph{depolarized} (thus, not pure) Scrooge ensembles.

\section{Main results}

We  consider a  quantum  system associated  with a  (finite)
$n$-dimensional Hilbert  space $\mathcal{H}$, and  we denote
with $\operatorname{Lin}(\mathcal{H})$  the space  of linear
operators on $\mathcal{H}$. Quantum  states of such a system
are    represented   by    density   matrices    $\rho   \in
\operatorname{Lin}(\mathcal{H})$,          that          is,
positive-semidefinite ($\rho \ge  0$) unit-trace ($\Tr[\rho]
= 1$)  operators. Any  discrete quantum  ensemble of  such a
system is  represented by a family  of sub-normalized states
$\{ \rho_x \in \operatorname{Lin}(\mathcal{H}) \}$, that is,
$\rho_x  \ge  0$ for  any  $x$  and $\Tr\sum_x  \rho_x  =1$.
Equivalently, $\rho  := \sum_x  \rho_x$ is a  quantum state,
and we say that the states composing the ensemble average to
$\rho$.  Any  discrete quantum measurement on  such a system
is represented  by a POVM,  that is  a family $\{  \pi_y \in
\operatorname{Lin}(\mathcal{H})     \}$      of     positive
semi-definite   operators,  such   that   $\sum_y  \pi_y   =
\openone$,  where $\openone$  represents  the unit  element,
that is  the element  with probability  $1$ over  any state.
The  joint probability  distribution  of  outcome $y$  given
input $x$  is given  by the  Born rule,  that is  $p_{x,y} =
\Tr[\rho_x \pi_y]$.  In the continuous case, summations must
be replaced by integrals. In the following, we will consider
both discrete  and continuous  ensembles and POVMs,  and for
simplicity we  will adopt the discrete  notation wherever it
suffices.

The  accessible information~\cite{Wil13}  $A(\{ \rho_x  \})$
and the  informational power~\cite{DDS11} $W(\{ \pi_y  \} )$
are  operationally   defined  as   the  maximum   amount  of
information that  can be extracted from  ensemble $\{ \rho_x
\}$ and by POVM $\{ \pi_y \}$, respectively:
\begin{align*}
  A(\{ \rho_x \}) = \max_{\{\pi_y\}} I( \{ \Tr[\rho_x \pi_y]
  \}),\\ W(\{ \pi_y \}) = \max_{\{\rho_x\}} I( \{ \Tr[\rho_x
    \pi_y] \}),
\end{align*}
where  the  maxima are  over  any  POVM  $\{ \pi_y  \}$  and
ensemble $\{ \rho_x \}$, respectively, and $I(\{ p_{x,y}\})$
denotes  the mutual  information  of  the joint  probability
distribution $\{ p_{x,y }\}$, that is
\begin{align*}
  I(\{    p_{x,y}\})     :=    \sum_{x,y}     p_{x,y}    \ln
  \frac{p_{x,y}}{p_x p_y},
\end{align*}
where $\{  p_x := \sum_y p_{x,y}  \}$ and $\{ p_y  := \sum_x
p_{x,y} \}$ are the marginals of $\{ p_{x,y} \}$.

The accessible  information and the informational  power are
related by the following duality formula~\cite{DDS11}, which
holds for any POVM $\{ \pi_y \}$:
\begin{align}
  \label{eq:duality}
  W(\{  \pi_y \})  = \max_\rho  A\left(\{ \sqrt{\rho}  \pi_y
  \sqrt{\rho} \}\right),
\end{align}
where   the  maximum   is   over  any   state  $\rho$.   The
Jozsa-Robb-Wootters   lower   bound    on   the   accessible
information~\cite{JRW94} of  any ensemble $\{ \rho_x  \}$ is
given by
\begin{align}
  \label{eq:acclower}
  A(\{  \rho_x   \})  \ge   Q(\rho)  -   \sum_x  \Tr[\rho_x]
  Q\left(\frac{\rho_x}{\Tr[\rho_x]} \right),
\end{align}
where  $\rho :=  \sum_x  \rho_x$ and  $Q(\rho)$ denotes  the
subentropy~\cite{JRW94}  of   $\rho$  (usually   defined  by
Eq.~\eqref{eq:subentropy2},   although    here   we   regard
$Q(\rho)$ as  a particular  case of the  quantity $Q_A(\rho,
P)$ defined  by Eq.~\eqref{eq:gen_subent_ens}).   The Holevo
upper bound~\cite{Hol73, Hol01} on accessible information is
given by
\begin{align}
  \label{eq:accupper}
  A(\{  \rho_x   \})  \le   S(\rho)  -   \sum_x  \Tr[\rho_x]
  S\left(\frac{\rho_x}{\Tr[\rho_x]} \right),
\end{align}
where $S(\rho)$ denotes the Von Neumann entropy~\cite{Hol73}
of $\rho$.   It is  well-known~\cite{Hol73, Rus02}  that the
bound  in Eq.~\eqref{eq:accupper}  is tight  if and  only if
$\rho_x$'s are pairwise commuting.

The aim of  this work is to study lower  and upper bounds on
the  accessible  information  $A(\{   \rho_x  \})$  and  the
informational power  $W(\{ \pi_y  \})$ under  constraints on
the purity $P$ of states $\{ \rho_x\}$ and POVM elements $\{
\pi_y  \}$,  where  $P(X)   :=  \Tr[X^2]/\Tr[X]^2$  for  any
self-adjoint operator $X$.

\subsection{Minimum information under purity constraint}

Our  first  result  is  a  lower  bound  on  the  accessible
information and informational power. For fixed Hilbert space
dimension $n$, denote with $Q_A(\rho, P)$ the minimum of the
accessible information  $A(\{\rho_x\})$ of any  ensemble $\{
\rho_x\}$ averaging to state $\rho$ such that $P(\rho_x) \ge
P$ for any $x$, that is
\begin{align}
  \label{eq:gen_subent_ens}
  Q_A(\rho,  P)  := \min_{\substack{\{\rho_x\}\\P(\rho_x)  \ge
      P\\\sum_x \rho_x = \rho}} A( \{ \rho_x \}).
\end{align}
Analogously,  denote  with  $Q_W(P)$   the  minimum  of  the
informational power $W(\{  \pi_y \})$ of any  POVM $\{ \pi_y
\}$ such that $P(\{ \pi_y \}) \ge P$ for any $y$. That is,
\begin{align}
  \label{eq:gen_subent_povm}
  Q_W(P) := \min_{\substack{\{\pi_y\}\\P(\pi_y)  \ge P}} W( \{
  \pi_y \}).
\end{align}

If $P =  1$, the quantity $Q_A(\rho, 1)  =: Q(\rho)$ reduces
to the well-known  subentropy~\cite{JRW94}.  Notice that, by
definition,  the  subentropy  $Q(\phi)$ of  any  pure  state
$\phi$ is  zero.  Ref.~\cite{JRW94} shows that  $Q(\rho)$ is
attained by the $\rho$-distorted  Scrooge ensemble, that is,
the  ensemble  of pure  states  $\{  n \sqrt{\rho}  \phi_x^*
\sqrt{\rho}\}$, where $\{ \phi_x^* \}$ denotes the uniformly
(Haar) distributed  ensemble.  If  $\rho =  \sum_k \lambda_k
\ket{\lambda_k}   \!\!   \bra{\lambda_k}$   is  a   spectral
decomposition of $\rho$, in  the absence of null eigenvalues
and degeneracies one explicitly obtains the formula
\begin{align}
  \label{eq:subentropy2}
  Q(\rho)     =    -     \sum_k    \frac{\lambda_k^n     \ln
    \lambda_k}{\prod_{j \neq k} (\lambda_k - \lambda_j)}.
\end{align}
Limits must  be considered in  case of null  eigenvalues and
degeneracies.   The formula~(\ref{eq:subentropy2})  is often
used to define the subentropy. The following expressions for
$Q_A(\rho,1)$~\cite{JRW94} and $Q_W(1)$~\cite{DBO14} follow
\begin{align}
  \label{eq:acclowerwc}
  \max_\rho Q_A(\rho, 1) =   Q_W(1) = \ln n - \Sigma_n.
\end{align}
Here and in  the following we set  $\Sigma_k := \sum_{j=2}^k
1/j$.

Our   first    main   result   consists    of   generalizing
Eq.~\eqref{eq:acclowerwc} to the case of arbitrary purity $P
\in [1/n, 1]$.
\begin{thm}[Lower bound under purity constraint]
  \label{thm:lower}
  One has
  \begin{align*}
    & \max_\rho Q_A(\rho, P) \\ \ge &  Q_W(P) \\ = & \ln n -
    \sum_{k=2}^n  {n \choose  k}  \frac{a^k \left(  \ln a  -
      \Sigma_k \right)}{(b-a)^{k-1}} +  \frac{b^n \left( \ln
      b - \Sigma_n \right)}{(b-a)^{n-1}},
  \end{align*}
  where  $a   :=  (1-\epsilon)/n$  and  $b   :=  \epsilon  +
  (1-\epsilon)/n$,   with   $\epsilon   :=  \sqrt{(n   P   -
    1)/(n-1)}$.  The  quantity $Q_W(P)$  is attained  by the
  $\epsilon$-depolarized     Scrooge      POVM     $\{     n
  \mathcal{D}_\epsilon (\phi_y^*) \}$.
\end{thm}

Here    and    in    the   following    we    denote    with
$\mathcal{D}_\epsilon$,  and we  call ``depolarizing  map,''
the positive (but not completely--positive) linear map given
by
\begin{align*}
  \mathcal{D}_\epsilon(\rho)  :=  \epsilon   \rho  +  (1-\epsilon)
  \Tr[\rho]  \frac{\openone}n, \qquad  -\frac1{n-1} \le
  \epsilon \le 1.
\end{align*}
Notice that the  above map is self-dual with  respect to the
trace: in other words, it acts on states and measurements in
the  same  way.   Also, the  map  $\mathcal{D}_\epsilon$  is
completely positive for $-(n^2-1)^{-1}  \le \epsilon \le 1$,
as  shown  in  Ref.~\cite{Kin03},  and  coincides  with  the
\textit{depolarizing channel} for $0 \le \epsilon \le 1$.

\subsection{Maximum information under purity constraint}

Our  second  result is  an  upper  bound on  the  accessible
information and informational power. For fixed Hilbert space
dimension $n$, denote with $S_A(\rho, P)$ the maximum of the
accessible information  $A(\{\rho_x\})$ of any  ensemble $\{
\rho_x \}$  averaging to  state $\rho$ such  that $P(\rho_x)
\le P$ for any $x$, that is
\begin{align}
  \label{eq:gen_ent_ens}
  S_A(\rho, P)  := \max_{\substack{\{\rho_x\}\\P(\rho_x) \le
      P\\\sum_x \rho_x = \rho}} A( \{ \rho_x \}).
\end{align}
Analogously,  denote  with  $S_W(P)$   the  maximum  of  the
informational power $W(\{\pi_y\})$ of any POVM $\{ \pi_y \}$
such that $P(\pi_y) \le P$ for any $y$, that is
\begin{align*}
  S_W(P) :=  \max_{\substack{\{\pi_y\}\\P(\pi_y) \le  P}} W(
  \{ \pi_y \}).
\end{align*}

If $P =  1$, the quantity $S_A(\rho, P)  =: S(\rho)$ reduces
to  the well-known  Von  Neumann entropy.   Notice that,  by
definition, the  entropy $S(\phi)$ of any  pure state $\phi$
is zero.  It is well-known that $S(\rho)$ is attained by the
ensemble given by the  spectral decomposition of $\rho$, and
is given by
\begin{align}
  \label{eq:entropy}
  S(\rho) = - \Tr\left[ \rho \log \rho \right].
\end{align}
The formula~\eqref{eq:entropy}  is often used to  define the
entropy.   The following  expressions for  $S_A(\rho,1)$ and
$S_W(1)$~\cite{DBO14} follow
\begin{align}
  \label{eq:gen_entropy_pure}
  S_W(1) = \max_\rho S_A(\rho, 1) = \ln n.
\end{align}

Our   second   main    result   consists   of   generalizing
Eq.~\eqref{eq:gen_entropy_pure}  to  the case  of  arbitrary
purity $P \in [1/n, 1]$.

\begin{thm}[Upper bound under purity constraint]
  \label{thm:upper}
  One has
  \begin{align*}
    S_W(P) \ge \max_\rho S_A(\rho, P) = \ln n +
    \floor{P^{-1}} a \ln a + b \ln b,
  \end{align*}
  where    $a     :=    (1     +    \sqrt{(P     \alpha    -
    1)/\floor{P^{-1}}})/\alpha$    and    $b   :=    (1    -
  \sqrt{\floor{P^{-1}}  (P   \alpha  -   1)})/\alpha$,  with
  $\alpha :=  \floor{P^{-1}} +  1$. The  quantity $\max_\rho
  S_A(\rho, P)$ is  attained by any ensemble  $\{ \rho_x \}$
  of $n$ states such that $\rho_x = a \ket{x}\!  \!\bra{x} +
  b \sum_{k \neq x} \ket{k}\!\!\bra{k}$ for any $x$, for any
  orthonormal basis $\{ \ket{k} \}$.
\end{thm}

The      results     of      Theorem~\ref{thm:lower}     and
Theorem~\ref{thm:upper}        are        depicted        in
Fig.~\ref{fig:tradeoff}.

\begin{figure}[h!]
  \begin{overpic}[width=\columnwidth]{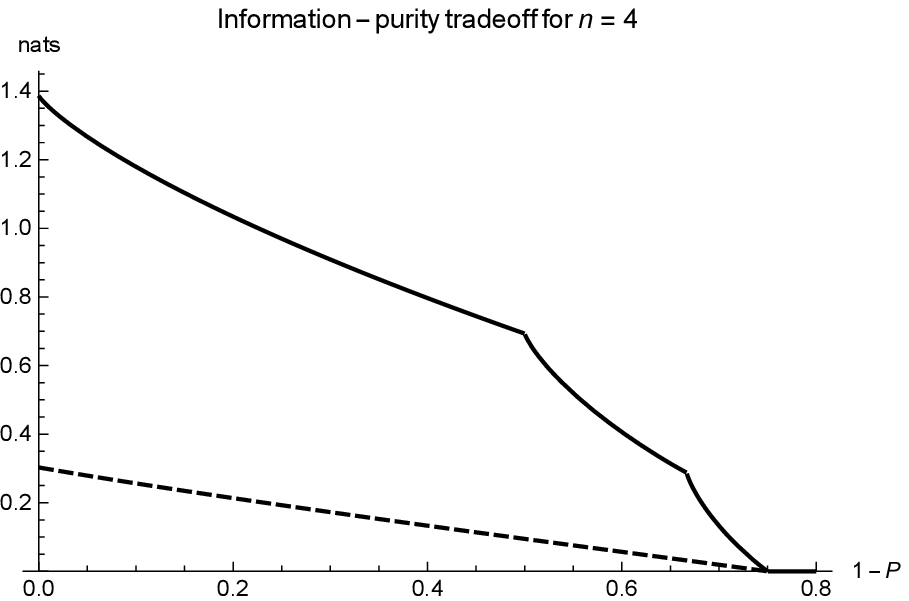}
  \end{overpic}
  \vphantom{X}
  
  \vphantom{X}
  \begin{overpic}[width=\columnwidth]{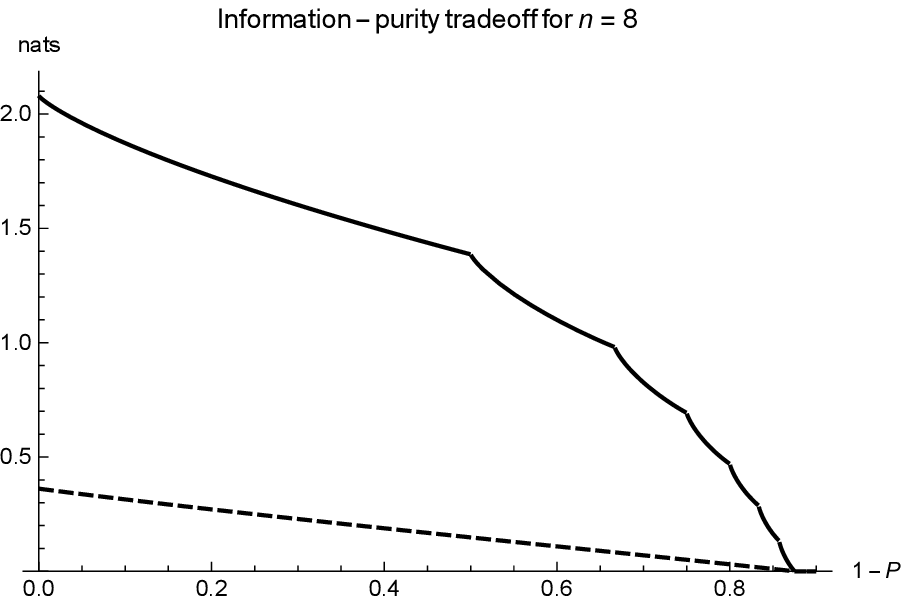}
  \end{overpic}
  \caption{Tradeoff between information  and impurity $1-P$,
    with $P \in [1/n, 1]$, for quantum ensembles and quantum
    measurements, for different values of the dimension $n$.
    The     quantity      $Q_W(P)$,     as      given     by
    Theorem~\ref{thm:lower},  is  represented by  the  lower
    dashed curve.   The maximum  value attained at  $P=1$ is
    $\ln n  - \Sigma_n$.  The quantity  $\max_\rho S_A(\rho,
    P)$, as given by Theorem~\ref{thm:upper}, is represented
    by the upper solid curve.  Notice that, as a consequence
    of  the dependence  on $\floor{P}^{-1}$  in Theorem~(2),
    the   quantity   $\max_\rho   Q_A(\rho,  P)$   has   $n$
    singularities, one for  each $P = 1/k$,  for any integer
    $1 \le k \le n$.}
  \label{fig:tradeoff}
\end{figure}

\setcounter{thm}{0}

\section{Proofs}

\subsection{Minimum     information    under     purity
  constraint}

The     aim    of     this    section     is    to     prove
Theorem~\ref{thm:lower}.  Notice   that  by   replacing  the
maximum over  $\rho$ with $\rho =  \openone/n$ in $\max_\rho
Q_A(\rho, P)$ one immediately has
\begin{align*}
  \max_\rho  Q_A  \left(  \rho,  P \right)  \ge  Q_A  \left(
  \frac{\openone}n, P \right).
\end{align*}
Since  Eq.~\eqref{eq:duality},  with   the  choice  $\rho  =
\openone/n$,   implies   that   \[  W(\{   \pi_y   \})   \ge
A\left(\left\{ \frac{\pi_y}n \right\}\right),
\] one
immediately has
\begin{align*}
  Q_W \left(  P \right)  \ge Q_A \left(  \frac{\openone}n, P
  \right).
\end{align*}
In other  words, both  quantities $\max_\rho Q_A  (\rho, P)$
and  $Q_W  (P)$  are  lower bounded  by  the  same  quantity
$Q_A(\openone / n, P)$.

In turn,  this common  lower bound can  be lower  bounded by
Eq.~\eqref{eq:acclower}.      Recalling~\cite{JRW94}    that
$Q(\openone/n) = \ln n - \Sigma_n$, one has
\begin{align}
  \label{eq:acc_low_bnd_max_mix}
  & Q_A \left( \frac{\openone}n, P \right) \nonumber\\ \ge &
  \ln n - \Sigma_n - \max_{\substack{\{\rho_x\} \\ P(\rho_x)
      \ge P\\\sum_x \rho_x = \openone/n}} \sum_x \Tr[\rho_x]
  Q \left( \frac{\rho_x}{\Tr[\rho_x]} \right).
\end{align}
Let  us   consider  the   last  term   in  the   r.h.s.   of
Eq.~\eqref{eq:acc_low_bnd_max_mix}.    Since  relaxing   the
constraint $\sum_x \rho_x = \openone  / n$ can only increase
the maximum and  the maximum of the average over  $x$ is not
larger than the largest element, one has
\begin{align*}
  \max_{\substack{\{\rho_x\}\\P(\rho_x) \ge P\\ \sum_x\rho_x
      =    \openone/n}}     \sum_x    \Tr[\rho_x]    Q\left(
  \frac{\rho_x}{\Tr[\rho_x]}       \right)       \le       &
  \max_{\substack{\rho\\P(\rho)\ge P}} Q \left(\rho\right).
\end{align*}
By         replacing          this         result         in
Eq.~\eqref{eq:acc_low_bnd_max_mix} one obtains
\begin{align}\label{eq:common-lower}
    Q_A\left(  \frac{\openone}n,  P  \right)  \ge  \ln  n  -
    \Sigma_n  -  \max_{\substack{\rho  \\  P(\rho)  \ge  P}}
    Q(\rho).
\end{align}

Hence, it suffices to compute  the maximum of the subentropy
$Q(\rho)$,  under the  constraint $P(\rho)  \ge P$.   First,
notice  that, without  loss  of  generality, the  constraint
$P(\rho) \ge P$ can be replaced with $P(\rho) = P$, that is
\begin{align*}
  \max_{\substack{\rho  \\
  P(\rho) \ge P}} Q(\rho) =   \max_{\substack{\rho  \\
  P(\rho) = P}} Q(\rho).
\end{align*}

Indeed, for any state $\rho$  such that $P(\rho) > P$, there
exists a value of $\epsilon$ such that the depolarized state
$\mathcal{D}_\epsilon(\rho)$ is such that $P(\rho) = P$, and
$Q(\mathcal{D}_\epsilon(\rho))  >  Q(\rho)$.   This  follows
from  the fact~\cite{DDJB14}  that $Q(\rho)$  is concave  in
$\rho$ and  maximized by $\rho  = \openone/n$, and  from the
fact that  $P(\rho)$ is  convex in  $\rho$ and  minimized by
$\rho         =        \openone/n$,         and        hence
$Q(\mathcal{D}_\epsilon(\rho))$                          and
$P(\mathcal{D}_\epsilon(\rho))$ are monotonically increasing
and decreasing in $\epsilon$, respectively.

Hence, in the following Lemma  we compute the maximum of the
subentropy $Q(\rho)$  under constraint  $P(\rho) =  P$.
\begin{lmm}
  \label{lmm:maxsub}
  The  maximum of  the subentropy  $Q(\rho)$ over  any state
  $\rho$ with purity $P(\rho) = P$, for any $P$, is attained
  by      any     $\epsilon$-depolarized      pure     state
  $\mathcal{D}_\epsilon(\phi)$, with
  \begin{align*}
    \epsilon = \sqrt{\frac{n P - 1}{n-1}}.
  \end{align*}
  Explicitly one has
  \begin{align*}
    & \max_{\substack{\rho \\  P(\rho) = P}} Q(\rho)  \\ = &
    \sum_{k=2}^n  {n \choose  k}  \frac{a^k \left(  \ln a  -
      \Sigma_k \right)}{(b-a)^{k-1}} -  \frac{b^n \left( \ln
      b - \Sigma_n \right)}{(b-a)^{n-1}} - \Sigma_n,
  \end{align*}
  where    $a$   and    $b$   are    the   eigenvalues    of
  $\mathcal{D}_\epsilon(\phi)$  with multiplicity  $n-1$ and
  $1$, respectively, that is
  \begin{align*}
    \begin{cases}
      a   :=   \frac{1-\epsilon}n,\\   b   :=   \epsilon   +
      \frac{1-\epsilon}n.
    \end{cases}
  \end{align*}
\end{lmm}

\begin{proof}
  We discuss here  a sketch of our proof,  which is formally
  provided in the Appendix.  Our  proof is based on a result
  of  Ref.~\cite{JM15},   where  the  maximization   of  the
  subentropy  was  considered  under  a  constraint  on  the
  symmetric polynomial  of degree  two.  We first  show that
  such a  constraint is  equivalent to a  purity constraint,
  and  hence  the  same  state $\rho$  is  optimal  for  the
  optimization  problem considered  here.  Then,  we compute
  the accessible  information $Q(\rho)$  of such  an optimal
  state, a  non-trivial task  given the $n-1$  degeneracy of
  its spectrum and hence the impossibility to directly apply
  Eq.~\eqref{eq:subentropy2}.   The explicit  calculation is
  carried  out  in  two  equivalent ways:  by  means  of  an
  integral  representation~\cite{JRW94}  of the  subentropy,
  and    by    means    of    a    formula    for    divided
  differences~\cite{SSJ17}.
\end{proof}

Applying        Lemma~\ref{lmm:maxsub}         to        the
bound~(\ref{eq:common-lower}),  we can  now lower  bound the
two quantities of interest as follows:
\begin{align}
  \label{eq:acclowerwcbound}
  & \max_\rho  Q_A \left( \rho, P  \right) \nonumber\\ \ge &  \ln n -
  \sum_{k=2}^n  {n  \choose  k}  \frac{a^k \left(  \ln  a  -
    \Sigma_k \right)}{(b-a)^{k-1}} +  \frac{b^n \left( \ln b
    - \Sigma_n \right)}{(b-a)^{n-1}}\;,\\
  \label{eq:infolowerwcbound}
  &  Q_W\left(  P   \right)  \nonumber\\  \ge  &   \ln  n  -
  \sum_{k=2}^n  {n  \choose  k}  \frac{a^k \left(  \ln  a  -
    \Sigma_k \right)}{(b-a)^{k-1}} +  \frac{b^n \left( \ln b
    - \Sigma_n \right)}{(b-a)^{n-1}}\;.
\end{align}

We  prove   now  the  tightness   of  the  lower   bound  in
Eq.~\eqref{eq:infolowerwcbound}.   In Ref.~\cite{BDS16},  in
the context of mixed $t$-designs, the accessible information
of   the   $\epsilon$-depolarized   Scrooge   ensemble   $\{
\mathcal{D}_\epsilon(\phi_x^*)  \}$  and  the  informational
power  of  the  $\epsilon$-depolarized Scrooge  POVM  $\{  n
\mathcal{D}_\epsilon(\phi_y^*)  \}$  were   derived  for  $0
\le\epsilon \le 1$.   We generalize that result  to the case
$-(n-1)^{-1} \le \epsilon \le 1$.

To  this   aim,  we  generalize   an  upper  bound   to  the
informational power derived in Ref.~\cite{BDS16} to the case
of  accessible information.   We start  by noticing  that by
definition
\begin{align*}
  & I(\{  \Tr[\rho_x \pi_y] \})  \\ =  & \ln n  + \sum_{x,y}
  \Tr[\rho_x      \pi_y]     \ln\left(      \frac{\Tr[\rho_x
      \pi_y]}{\Tr[\rho_x]   \Tr[\pi_y]}  \right) \\ &  -  \sum_y
  \Tr[\rho     \pi_y]     \ln\left(     n     \frac{\Tr[\rho
      \pi_y]}{\Tr[\pi_y]} \right).
\end{align*}
Since both $\{ \Tr[\pi_y]/n \}$  and $\{ \Tr[\rho \pi_y] \}$
are probability  distributions, the last term  in the r.h.s.
is  the relative  entropy  $D(\{ \Tr[\rho  \pi_y]  \} ||  \{
\Tr[\pi_y]/n \})$, which is  non-negative, and null if $\rho
=  \openone/n$. Hence,  disregarding  the last  term in  the
r.h.s. and setting $\eta(x) := -x \ln x$, one has
\begin{align*}
  & I(\{ \Tr[\rho_x  \pi_y] \}) \\ \le & \ln  n + \sum_{x,y}
  \Tr[\rho_x      \pi_y]     \ln\left(      \frac{\Tr[\rho_x
      \pi_y]}{\Tr[\rho_x] \Tr[\pi_y]} \right) \\ = & \ln n +
  \sum_{x,y}      \Tr[\rho_x]     \Tr[\pi_y]\frac{\Tr[\rho_x
      \pi_y]}{\Tr[\rho_x]        \Tr[\pi_y]}       \ln\left(
  \frac{\Tr[\rho_x      \pi_y]}{\Tr[\rho_x]      \Tr[\pi_y]}
  \right)\\   =   &   \ln   n   -   \sum_{x,y}   \Tr[\rho_x]
  \Tr[\pi_y]\ \eta\left(\frac{\Tr[\rho_x\pi_y]}{\Tr[\rho_x]\Tr[\pi_y]}
  \right),
\end{align*}
which  is nicely  symmetric  in the  ensemble  and the  POVM
(notice that in the  denominator we have $\Tr[\pi_y]$ rather
than $\Tr[\rho \pi_y]$). Hence, we can use it to upper bound
both the accessible information  and the informational power
in the  same way.  Notice  also that the argument  of $\eta$
does not depend on the traces of $\rho_x$ and $\pi_y$: these
can be  rescaled at will  without changing the value  of the
ratio within parentheses. Thus we  can recast the problem as
an optimization over a single normalized state, as follows.

By definition, the accessible  information is the maximum of
the mutual information over all POVMs, hence
\begin{align*}
  & A(  \{\rho_x\} )  \\ \le  & \ln  n -  n \min_{\{\pi_y\}}
  \sum_{x,y}  \Tr[\rho_x]  \frac{\Tr[\pi_y]}n  \  \eta\left(
  \frac{\Tr[\rho_x \pi_y]}{\Tr[\rho_x]\Tr[\pi_y]} \right).
\end{align*}
In the above  equation, we introduced a factor  $n$, so that
the    coefficient   $\Tr[\pi_y]/n$    is   a    probability
distribution.  Hence,  the minimum  over $\{\pi_y\}$  of the
average over $y$ is not  less than the global minimum, i.e.,
it can be bounded as follows:
\begin{align}
  \label{eq:acc_up_bnd_1}
  A(\{\rho_x\})   \le   \ln   n   -   n   \min_\phi   \sum_x
  \Tr[\rho_x]\          \eta\left(          \frac{\Tr[\rho_x
      \phi]}{\Tr[\rho_x]} \right),
\end{align}
where  now the  minimum is  taken over  a single  normalized
state $\phi$ (which  can be chosen pure,  although this does
not matter  at this point).   Notice that equality  holds if
$\sum_x \rho_x = \openone/n$ and $\openone/n$ belongs to the
convex hull of  the set of states attaining  the minima over
$\phi$. The former condition  is sufficient for the relative
entropy $D(\{ \Tr[\rho \pi_y] \}  || \{ \Tr[\pi_y]/n \})$ to
be  zero,  as  discussed   before.   The  latter  condition,
instead,  is necessary  and  sufficient for  the r.h.s.   of
Eq.~\eqref{eq:acc_up_bnd_1}   to   be  equivalent   to   the
r.h.s. of the previous equation.

Along exactly the same lines, by definition of informational
power, one has
\begin{align*}
  & W(\{  \pi_y \})  \\ \le  & \ln  n -  n \min_{\{\rho_x\}}
  \sum_{x,y}  \Tr[\rho_x] \frac{\Tr[\pi_y]}n  \ \eta  \left(
  \frac{\Tr[\rho_x \pi_y]}{\Tr[\rho_x] \Tr[\pi_y]}\right).
\end{align*}
Again,  since $\Tr[\rho_x]$  is  a probability  distribution
over $x$,  the minimum  over $\{ \rho_x  \}$ of  the average
over  $x$ is  lower bounded  by  the minimum  over a  single
normalized state $\phi$ as follows
\begin{align}
  \label{eq:info_up_bnd_1}
  W( \{ \pi_y \} )  \le \ln  n -  n \min_\phi  \sum_y \frac{\Tr[\pi_y]} n \    \eta\left(    \frac{\Tr[\pi_y \phi]}{\Tr[\pi_y]} \right),
\end{align}
with equality if $\openone/n$ belongs  to the convex hull of
the set of states attaining the minima over $\phi$.

We  compute the  bounds  in Eq.~\eqref{eq:acc_up_bnd_1}  and
Eq.~\eqref{eq:info_up_bnd_1} for the  depolarized version of
the uniformly distributed pure ensemble, that is $\{\rho_x =
\mathcal{D}_\epsilon(\phi_x^*)\}$,  and for  the depolarized
version of the uniformly  distributed rank-one POVM, that is
$\{n  \mathcal{D}_\epsilon(\phi_y^*)\}$,   respectively.  We
also show  that, in these  cases, the bounds are  tight.  To
these aims,  first notice that  the summation in  the r.h.s.
of              Eq.~\eqref{eq:acc_up_bnd_1}              and
Eq.~\eqref{eq:info_up_bnd_1}, which  are identical  in form,
must be replaced in these  cases by an integral over uniform
measure $\dif \mu_x$. Since, by direct calculation,
\begin{align*}
  \Tr[\mathcal{D}_\epsilon   (\phi^*_g)    \phi]   =   (b-a)
  \left|\braket{\phi | \phi_g^*}\right|^2 + a \ ,
\end{align*}
one has, setting $g(x) := (b-a) x + a$,
\begin{align}
  \label{eq:integral}
  & \min_\phi \int \dif \mu_x \ \braket{\phi_x^* | \phi_x^*}
  \eta   \left(  \frac{\Tr[\mathcal{D}_\epsilon   (\phi^*_x)
      \phi]}{\braket{\phi_x^*    |     \phi_x^*}}    \right)
  \nonumber\\    =    &    \min_\phi   \int    \dif    \mu_x
  \  \braket{\phi_x^*  |  \phi_x^*}   \eta  \circ  g  \left(
  \frac{\left|\braket{\phi                                 |
      \phi_x^*}\right|^2}{\braket{\phi_x^*    |   \phi_x^*}}
  \right) \ .
\end{align}

Due  to  unitary  invariance,  the minimum  over  $\phi$  is
independent  of  $\phi$, so  in  the  following $\phi$  will
denote an arbitrarily chosen  pure state.  Hence, the bounds
in              Eq.~\eqref{eq:acc_up_bnd_1}              and
Eq.~\eqref{eq:info_up_bnd_1} are tight.

To    compute    the    integral   in    the    r.h.s.    of
Eq.~\eqref{eq:integral}, we resort  to the following result,
proved  in Refs.~\cite{Jon91a}  and~\cite{Jon91b}.  For  any
integrable function $f$ one has
\begin{align}
  \label{eq:jones}
  & \int  \dif \mu_x  \braket{\phi_x^* | \phi_x^*}  f \left(
  \frac{\left|\braket{\phi                                 |
      \phi_x^*}\right|^2}{\braket{\phi_x^*    |   \phi_x^*}}
  \right)  \nonumber\\ =  &  (n-1)!   \left[ [f]^{n-1}(1)  -
    \sum_{k=2}^n \frac{[f]^{k-1}(0)}{(n-k)!}  \right],
\end{align}
where  $\{  [f]^m  \}_{m=1}^{n-1}$ represents  a  choice  of
$m$-degree antiderivatives of $f$,  namely $[f] := \int \dif
x f(x)$  and $[f]^m := [  [f]^{m-1} ]$.  Of course,  for any
choice  of $[f]^{m-1}$,  one  has that  $[f]^m$ is  uniquely
defined  up  to  a  constant,  but  Eq.~\eqref{eq:jones}  is
independent  of  such   a  choice  (see  Refs.~\cite{Jon91a,
  Jon91b}).

It was  also shown in  Refs.~\cite{Jon91a} and~\cite{Jon91b}
that
\begin{align*}
  [\eta]^m  =  -  \frac{x^{m+1}}{(m+1)!}  \left(  \log  x  -
  \Sigma_{m+1} \right) \ ,
\end{align*}
and,  given  that  $g$  is an  affine  function,  by  direct
computation one immediately has
\begin{align}
  \label{eq:antiderivatives}
  [\eta \circ g]^m = \frac1{(b-a)^m} [\eta]^m \circ g \; .
\end{align}
Since in our  case one has $f = \eta  \circ g$, by replacing
Eq.~\eqref{eq:antiderivatives} into  Eq.~\eqref{eq:jones} we
obtain the accessible information $A(\{ \mathcal{D}_\epsilon
(\phi_x^*)\})$ of  the depolarized version of  the uniformly
distributed  ensemble  $\{\mathcal{D}_\epsilon(\phi_x^*)\}$,
and      the      informational      power      $W(\{      n
\mathcal{D}_\epsilon(\phi_y^*)   \})$  of   the  depolarized
version  of the  uniformly distributed  rank-one POVM  $\{ n
\mathcal{D}_\epsilon(\phi_y)\}$, as follows
\begin{align*}
  & A( \{ \mathcal{D}_\epsilon(\phi_x^*) \})  \\ = & W( \{ n
  \mathcal{D}_\epsilon  (\phi_y^*)  \})  \\  =  &  \ln  n  -
  \sum_{k=2}^n  {n  \choose  k}  \frac{a^k \left(  \ln  a  -
    \Sigma_k \right)}{(b-a)^{k-1}} +  \frac{b^n \left( \ln b
    - \Sigma_n \right)}{(b-a)^{n-1}} \; ,
\end{align*}
which proves the tightness of the bound on the informational
power   in  Eq.\eqref{eq:infolowerwcbound}   (but  not   the
tightness  of the  bound  on the  accessible information  in
Eq.\eqref{eq:acclowerwcbound},  given that  it is  a maximin
problem).

Summarizing, we have the following first main result.

\begin{thm}[Lower bound under purity constraint]
  One has
  \begin{align*}
    & \max_\rho Q_A(\rho, P) \\ \ge &  Q_W(P) \\ = & \ln n -
    \sum_{k=2}^n  {n \choose  k}  \frac{a^k \left(  \ln a  -
      \Sigma_k \right)}{(b-a)^{k-1}} +  \frac{b^n \left( \ln
      b - \Sigma_n \right)}{(b-a)^{n-1}},
  \end{align*}
  where  $a   :=  (1-\epsilon)/n$  and  $b   :=  \epsilon  +
  (1-\epsilon)/n$,   with   $\epsilon   :=  \sqrt{(n   P   -
    1)/(n-1)}$.  The  quantity $Q_W(P)$  is attained  by the
  $\epsilon$-depolarized     Scrooge      POVM     $\{     n
  \mathcal{D}_\epsilon (\phi_y^*) \}$.
\end{thm}

Theorem~\ref{thm:lower} sheds  new light  on the  problem of
the tightness of the  Jozsa-Robb-Wootters lower bound on the
accessible information  in Eq.~\eqref{eq:acclower}.  Indeed,
from  Eq.~\eqref{eq:acc_low_bnd_max_mix} it  follows that  a
sufficient condition for tightness  is that the ensemble $\{
\rho_x  \}$ is  the $\epsilon$-depolarized  Scrooge ensemble
$\{  \mathcal{D}_\epsilon (\phi_x^*)  \}$,  for  any $0  \le
\epsilon \le 1$.  This generalizes the previously known fact
that  the  bound  in Eq.~\eqref{eq:acclower}  is  tight  for
$\epsilon = 1$.

\subsection{Maximum     information    under     purity
  constraint}

The aim of this section is to prove Theorem~\ref{thm:upper}.
By  applying  Eq.~\eqref{eq:accupper}  and using  the  bound
$S(\rho) \le \ln n$ we have
\begin{align*}
  \max_\rho S_A  \left( \rho, P  \right) \le \ln n  - \min_{
    \substack{\{  \rho_x  \}\\P(\rho_x)   \le  P}  }  \sum_x
  \Tr[\rho_x] S\left(\frac{\rho_x}{\Tr[\rho_x]} \right),
\end{align*}
which is  tight if and only  if the minimum over  $\{ \rho_x
\}$ is attained by an ensemble of commuting states averaging
to  the maximally  mixed  state. Since  the  minimum of  the
average of $S(\rho_x / \Tr[\rho_x])$ is not smaller that the
minimum of $S(\rho)$, one has
\begin{align*}
    \max_\rho  S_A  \left(  \rho,  P \right)  \le  \ln  n  -
    \min_{\substack{\rho\\P(\rho)   \le    P}}   S\left(\rho
    \right),
\end{align*}
which  is tight  if and  only if  the maximally  mixed state
belongs to  the convex hull  of the set of  states attaining
the minimum over $\rho$.

Hence, in the following we address the problem of minimizing
the  Von  Neumann  entropy  under  an  upper  bound  on  the
purity. First, notice  that since $P(\rho) \le  P$ defines a
convex set and $S(\rho)$ is concave, the minimum is attained
on the boundary, that is
\begin{align*}
  \min_{\substack{\rho\\P(\rho) \le  P}} S\left(\rho \right)
  = \min_{\substack{\rho\\P(\rho) = P}} S\left(\rho \right).
\end{align*}

In  Ref.~\cite{HT01}, the  maximum  and minimum  of the  Von
Neumann entropy  under an equality constraint  on the purity
(therein  referred  to as  the  index  of coincidence  of  a
classical  probability  distribution)  were derived  with  a
topological approach.  We also notice that analogous results
were  discussed  in   Ref.~\cite{WNGKMV03}  to  characterize
maximally  entangled   states  for   given  purity   of  the
marginals. In the following Lemma, we provide a simple proof
of a partial result  of Ref.~\cite{HT01}, that we generalize
to    the     case    of    arbitrary     R\'enyi    entropy
$H_\alpha(\vec\lambda)   :=   (1-\alpha)^{-1}   \ln   \sum_k
\lambda_k^\alpha$.   The  case  of Von  Neumann  entropy  is
recovered    since    $S(\rho)    =    \lim_{\alpha\to    1}
H_\alpha(\vec\lambda)$,  where  $\rho  :=  \sum_k  \lambda_k
\ket{\lambda_k}  \!  \!   \bra{\lambda_k}$,  and the  purity
constraint $P(\rho) = P$ becomes $|\vec\lambda|_2^2 = P$.

\begin{lmm}
  \label{lmm:maxminent}
  Under constraints $\vec\lambda  \ge 0$, $|\vec\lambda|_1 =
  1$,   and  $|\vec\lambda|_2^2   =  P$,   the  extrema   of
  $H_\alpha(\vec\lambda)$  are attained  by a  $\vec\lambda$
  with at most two different non-null eigenvalues, that is
  \begin{align*}
    \vec\lambda =  \left( a_\pm, \dots, a_\pm,  b_\pm, \dots
    b_\pm, 0, \dots 0 \right),
  \end{align*}
  where  $(a_+, b_+)$  and  $(a_-, b_-)$  are  the only  two
  assignments   that  satisfy   the  constraints,   and  are
  explicitly given by
  \begin{align}
    \label{eq:a_pm}
    a_\pm & := \frac{1\pm\sqrt{\frac{n_b}{n_a}\left( P(n_a +
        n_b) -1 \right)}}{n_a + n_b},\\
    \label{eq:b_pm}
    b_\pm & := \frac{1\mp\sqrt{\frac{n_a}{n_b}\left( P(n_a +
        n_b) -1 \right)}}{n_a + n_b},
  \end{align}
  where  $n_a$ and  $n_b$  denote  the multiplicity  of
  $a_\pm$ and $b_\pm$, respectively. Explicitly one has
  \begin{align*}
    \min_{\substack{\vec\lambda\ge0\\|\vec\lambda|_1=1\\|\vec\lambda|_2^2
        = P}} H_\alpha\left(\vec\lambda \right) = \min_{n_a,
      n_b,  \pm} \left[  \frac{1}{1-\alpha}  \ln \left(  n_a
      a_\pm^\alpha     +     n_b    b_\pm^\alpha     \right)
      \right],\\ \max_{\substack{\vec\lambda\ge0\\|\vec\lambda|_1=1\\|\vec\lambda|_2^2
        = P}} H_\alpha\left(\vec\lambda \right) = \max_{n_a,
      n_b,  \pm} \left[  \frac{1}{1-\alpha}  \ln \left(  n_a
      a_\pm^\alpha + n_b b_\pm^\alpha \right) \right].
  \end{align*}
\end{lmm}

\begin{proof}
  We discuss here  a sketch of our proof,  which is formally
  provided  in  the Appendix.   First,  we  notice that  the
  equality and  inequality constrained optimizations  of the
  R\'enyi entropy  are equivalent to a  set of equality-only
  constrained  optimizations  in smaller  dimensions.   This
  allows  us to  successfully apply  the method  of Lagrange
  multipliers to solve such a set of optimization problems.
\end{proof}

We remark that Lemma~\ref{lmm:maxminent}  is in closed form,
because  it involves  a minimization  over $n_a$  and $n_b$,
non-negative integers such that $n_a + n_b \le n$.  However,
Ref.~\cite{HT01} provides  additional insight (for  the case
of Von Neumann entropy) since  such a minimization is solved
therein.  It was shown in  Ref.~\cite{HT01} that for the Von
Neumann entropy  $S(\rho)$ the minimum over  $n_a$ and $n_b$
is attained by $n_a = \floor{P^{-1}}$  and $n_b = 1$, and by
$a_+$, $b_+$.  So  we have the following upper  bound on the
accessible information
\begin{align*}
  \max_\rho Q_A \left(  \rho, P \right) \le \ln n  + n_a a_+
  \ln a_+ + b_+ \ln b_+,
\end{align*}
where $n_a =  \floor{P^{-1}}$ and $a_+$, $b_+$  are as given
by  Lemma~\ref{lmm:maxminent}.   Moreover,   this  bound  is
tight, since the maximally mixed state belongs to the convex
hull of  the set of  states obtained by considering  all the
permutations of  the eigenvalues  $\vec\lambda$ as  given by
Lemma~\ref{lmm:maxminent},   for   some  fixed   basis   $\{
\ket{\lambda_k} \}$. By taking the  same structure as a POVM
$\{ \pi_y \}$ one also has
\begin{align*}
  W(\{ \pi_y \}) = \ln n + n_a a_+ \ln a_+ +  b_+ \ln b_+. 
\end{align*}

Hence we have our second main result
\begin{thm}[Upper bound under purity constraint]
  One has
  \begin{align*}
    S_W(P) \ge \max_\rho S_A(\rho, P) = \ln n +
    \floor{P^{-1}} a \ln a + b \ln b,
  \end{align*}
  where    $a     :=    (1     +    \sqrt{(P     \alpha    -
    1)/\floor{P^{-1}}})/\alpha$    and    $b   :=    (1    -
  \sqrt{\floor{P^{-1}}  (P   \alpha  -   1)})/\alpha$,  with
  $\alpha :=  \floor{P^{-1}} +  1$. The  quantity $\max_\rho
  S_A(\rho, P)$ is  attained by any ensemble  $\{ \rho_x \}$
  of $n$ states such that $\rho_x = a \ket{x}\!  \!\bra{x} +
  b \sum_{k \neq x} \ket{k}\!\!\bra{k}$ for any $x$, for any
  orthonormal basis $\{ \ket{k} \}$.
\end{thm}

\section{Conclusion}

Known subentropy-like lower bounds on informational measures
(accessible information and  informational power) all assume
the optimization  to be restricted  to pure states  and POVM
elements.   In this  work,  we relaxed  this assumption,  by
regarding purity  as a resource, thus  recasting the problem
as  an  information-purity   tradeoff.   In  particular,  we
computed the minimum informational  power when the purity is
lower bounded,  and the maximum accessible  information when
the purity  is upper  bounded.  We  provided bounds  for the
other cases. We also discussed  the problem of the tightness
of  the   Jozsa-Robb-Wootters  lower  bound   on  accessible
information, giving new cases in which it is tight.

We conclude by discussing  some relevant open problems:
\begin{itemize}
\item  It is  still an  open problem  whether our  bounds in
  Theorem~\ref{thm:lower} and~\ref{thm:upper} are tight.
\item In  Lemma~\ref{lmm:maxsub} we  derived the  maximum of
  the subentropy under a  purity constraint; analogously, in
  Lemma~\ref{lmm:maxminent}  we  derived   the  maximum  and
  minimum  of  the  Von   Neumann  entropy  under  a  purity
  constraint.  It is still open  the problem of deriving the
  minimum of the subentropy  under a purity constraint.  One
  approach to this problem would involve extending the proof
  technique of Lemma 5 of Ref.~\cite{JM15}.
\item Here  we showed that the  Jozsa-Robb-Wootters bound in
  Eq.~\eqref{eq:acclower} is tight not  only for the Scrooge
  pure ensembles  and measurements, but also  when these are
  $\epsilon$-depolarized.  It  is still open the  problem of
  deriving  necessary  and  sufficient  conditions  for  the
  tightness of the subentropy lower bounds in general.
\item Closed  expressions for the quantities  $Q_A(\rho, P)$
  and $S_A(\rho,  P)$ are well-known  for the case $P  = 1$,
  for any $\rho$.  We  introduced closed expressions for any
  $P$,  in the  case $\rho  = \openone/n$.   Deriving closed
  expressions for $Q_A(\rho, P)$  and $S_A(\rho, P)$ for any
  $\rho$ and  $P$ is still  an open problem. In  this sense,
  there is still a trade-off in our current understanding of
  the relation between information and purity.
\end{itemize}

\setcounter{lmm}{0}

\section*{Appendix: proofs of the Lemmas}

Here         we        prove         Lemmas~\ref{lmm:maxsub}
and~\ref{lmm:maxminent}, that we recall for convenience.

\subsection{Maximum subentropy under purity constraint}

\begin{lmm}\label{lmm:maxsub-app}
  The  maximum of  the subentropy  $Q(\rho)$ over  any state
  $\rho$ with purity $P(\rho) = P$, for any $P$, is attained
  by      any     $\epsilon$-depolarized      pure     state
  $\mathcal{D}_\epsilon(\phi)$, with
  \begin{align}
    \label{eq:lambda2}
    \epsilon = \sqrt{\frac{n P - 1}{n-1}}.
  \end{align}
  Explicitly one has
  \begin{align*}
    & \max_{\substack{\rho \\  P(\rho) = P}} Q(\rho)  \\ = &
    \sum_{k=2}^n  {n \choose  k}  \frac{a^k \left(  \ln a  -
      \Sigma_k \right)}{(b-a)^{k-1}} -  \frac{b^n \left( \ln
      b - \Sigma_n \right)}{(b-a)^{n-1}} - \Sigma_n,
  \end{align*}
  where    $a$   and    $b$   are    the   eigenvalues    of
  $\mathcal{D}_\epsilon(\phi)$  with multiplicity  $n-1$ and
  $1$, respectively, that is
  \begin{align*}
    \begin{cases}
      a   :=   \frac{1-\epsilon}n,\\   b   :=   \epsilon   +
      \frac{1-\epsilon}n.
    \end{cases}
  \end{align*}
\end{lmm}

\begin{proof}
  Upon setting $\rho =:  \sum_k \lambda_k \ket{\lambda_k} \!
  \bra{\lambda_k}$ and $e_2(\rho) :=  \sum_{k < j} \lambda_k
  \lambda_j$, it has been proven~\cite{JM15} (see Property 7
  and Lemma 5 therein) that
  \begin{align*}
    \max_{\substack{\rho  \\ e_2(\rho) = E}} Q(\rho)
  \end{align*}
  is attained by  $\rho = \mathcal{D}_\epsilon(\phi)$, where
  $\phi$   is  any   pure  state   and  $\epsilon$   is  the
  non-negative parameter such that the constraint $e_2(\rho)
  = E$ is satisfied.  Since  by explicit computation one has
  $P(\rho)  = 1  - 2  e_2(\rho)$, the  maximum under  purity
  constraint    is     also    attained    by     $\rho    =
  \mathcal{D}_\epsilon(\phi)$   and    $\epsilon$   is   the
  non-negative parameter such that the constraint $P(\rho) =
  P$ is satisfied.  By explicit computation one has
  \begin{align*}
    P(\rho) = \frac{(n-1)\epsilon^2 + 1}n,
  \end{align*}
  hence Eq.~\eqref{eq:lambda2} immediately follows.

  In   order  to   compute  $Q(\mathcal{D}_\epsilon(\phi))$,
  Eq.~\eqref{eq:subentropy2} is unpractical  as the spectrum
  of  $\mathcal{D}_\epsilon(\phi)$ is  degenerate.  Here  we
  compute $Q(\mathcal{D}_\epsilon(\phi))$ using the integral
  representation of $Q(\rho)$  derived in Ref.~\cite{JRW94}.
  One has $Q(\rho) = G(\rho) - \Sigma_n$, where
  \begin{align*}
    G(\rho) := -n \int  \dif x \left( \sum_{k=1}^n \lambda_k
    x_k  \right)  \ln   \left(  \sum_{k=1}^n  \lambda_k  x_k
    \right),
  \end{align*}
  and       $\rho       =       \sum_{k=1}^n       \lambda_k
  \ket{\lambda_k}\!\!\bra{\lambda_k}$.  The integral is over
  the simplex of probabilities given  by $x_k \ge 0$ for any
  $k$ and $\sum_k x_k = 1$, that is
  \begin{align*}
    \int   \dif   x   :=   N   \int_0^1   \dif   x_1   \dots
    \int_0^{1-x_1\dots -x_{n-2}} \dif x_{n-1},
  \end{align*}
  where $N$ denotes a  normalization factor that was derived
  e.g. in Eq.~(A1), Appendix 1, of Ref.~\cite{JRW94}.

  For  the sake  of completeness,  let us  here compute  $N$
  again by iteratively applying the integration formula
  \begin{align*}
    \int_0^\beta  \dif  x  \left(  \beta  -  x  \right)^m  =
    \frac{\beta^{m+1}}{m+1},
  \end{align*}
  easily  obtained  by replacing  $t  :=  \beta -  x$,  thus
  eventually obtaining
  \begin{align*}
    \int_0^1  \dif  x_1 \dots  \int_0^{1-x_1\dots  -x_{k-1}}
    \dif x_k = \frac1{k!}\;, \qquad \forall k,
  \end{align*}
  Hence  the condition  $\int  \dif  x =  1$  requires $N  =
  (n-1)!$.

  To  compute  $G(\mathcal{D}_\epsilon(\phi))$, notice  that
  for $\rho = \mathcal{D}_\epsilon(\phi)$ one has $\lambda_k
  = a$ for $1 \le k \le  n-1$ and $\lambda_n = b$, with $a =
  (1-\epsilon)/n$ and  $b = \epsilon +  (1-\epsilon)/n$.  We
  set $c := b-a$.

  We   compute    now   $G(\mathcal{D}_\epsilon(\phi))$   by
  iteratively applying the integration formulas
  \begin{align*}
    & \int_0^{1-\alpha} \dif x (b - c(\alpha + x))^{m-1} \ln
    (b - c(\alpha + x)) \\ = & \frac{ (b - c\alpha)^m \left(
      \ln (b - c\alpha) - \frac1m \right) - a^m\left(\ln a -
      \frac1m\right)}{mc},
  \end{align*}
  easily derived  by substituting $t  := b - c(\alpha  - x)$
  and by partial integration, and
  \begin{align*}
    \int_0^{1-\alpha}  \dif x  (b -  c(\alpha +  x))^{m-1} =
    \frac{ (b - c\alpha)^m - a^m}{mc},
  \end{align*}
  easily derived by  substituting $t := b -  c(\alpha - x)$,
  thus eventually obtaining
  \begin{align*}
    G(\mathcal{D}_\epsilon(\phi)) =  \sum_{k=2}^n {n \choose
      k}    \frac{a^k    \left(    \ln    a    -    \Sigma_k
      \right)}{(b-a)^{k-1}}  -  \frac{b^n  \left(  \ln  b  -
      \Sigma_n \right)}{(b-a)^{n-1}},
  \end{align*}
  or equivalently
  \begin{align}
    \label{eq:sub}
    &   Q(\mathcal{D}_\epsilon(\phi))    \nonumber\\   =   &
    \sum_{k=2}^n  {n \choose  k}  \frac{a^k \left(  \ln a  -
      \Sigma_k \right)}{(b-a)^{k-1}} -  \frac{b^n \left( \ln
      b - \Sigma_n \right)}{(b-a)^{n-1}} - \Sigma_n.
  \end{align}

  An         alternative        is         to        compute
  $Q(\mathcal{D}_\epsilon(\phi))$  by   divided  differences
  (see Eq.~11 of Ref.~\cite{SSJ17}), in which case one has
  \begin{align}
    \label{eq:divdiff}
    Q(\mathcal{D}_\epsilon(\phi))      =      \frac1{(n-2)!}
    \frac{\partial^{n-2}}{\partial a^{n-2}} \left( \frac{a^n
      \log a}{b-a} - \frac{b^n \log b}{b - a}\right).
  \end{align}
  One immediately has
  \begin{align}
    \label{eq:diff1}
    \frac{\partial^{n-2}}{\partial  a^{n-2}} \frac{b^n  \log
      b}{b - a} = (n-2)! \frac{b^n \log b}{(b-a)^{n-1}}.
  \end{align}
  By applying the multinomial theorem, one also has
  \begin{align*}
    & \frac{\partial^{n-2}}{\partial a^{n-2}} \frac{a^n \log
      a}{b  - a}  \\ =  & \sum_{k_1  + k_2  + k_3  = n  - 2}
    \frac{(n-2)!}{k_1!   k_2!   k_3!}   \times \\  &  \times
    \left(   \frac{\partial^{k_1}}{\partial   a^{k_1}}   a^n
    \right)  \left( \frac{\partial^{k_2}}{\partial  a^{k_2}}
    \log  a  \right)  \left(  \frac{\partial^{k_3}}{\partial
      a^{k_3}} (b-a)^{-1} \right).
  \end{align*}
  Since of course
  \begin{align*}
    \frac{\partial^k}{\partial     a^k}      a^n     &     =
    \frac{n!}{(n-k)!}  a^{n-k},\\ \frac{\partial^k}{\partial
      a^k} \log a & =
    \begin{cases}
      \log a & \textrm{ if } k = 0,\\ (-1)^{k-1}(k-1)!a^{-k}
      & \textrm{ if } k > 0,
    \end{cases}\\
    \frac{\partial^k}{\partial    a^k}   (b-a)^{-1}    &   =
    k!(b-a)^{-(k+1)},
  \end{align*}
  one has
  \begin{align}
    \label{eq:diff2}
    & \frac{\partial^{n-2}}{\partial a^{n-2}} \frac{a^n \log
      a}{b  -  a}  \nonumber\\   =  &  (n-2)!   \sum_{k=2}^n
    \frac{a^k}{(b-a)^{k-1}}     \times     \nonumber\\     &
    \times\left[ {n \choose  k} \log a -  \sum_{j = 1}^{n-2}
      {n \choose k + j} \frac{(-1)^{j}}j \right].
  \end{align}
  Combining   Eq.~\eqref{eq:divdiff},  Eq.~\eqref{eq:diff1},
  and Eq.~\eqref{eq:diff2} one finally has
  \begin{align}
    \label{eq:sub2}
    &   Q(\mathcal{D}_\epsilon(\phi))    \nonumber\\   =   &
    \sum_{k=2}^n  {n \choose  k}  \frac{a^k \left(  \ln a  -
      \Sigma_k \right)}{(b-a)^{k-1}} -  \frac{b^n \left( \ln
      b  -  \Sigma_n   \right)}{(b-a)^{n-1}}  -  \Sigma_n  +
    \nonumber\\ & + \Sigma_n\left[(n-1)a + b -1\right].
  \end{align}
  which differs from Eq.~\eqref{eq:sub}  only by a term that
  vanishes    due   to    unit-trace.    The    r.h.s.    of
  Eq.~\eqref{eq:sub2} should replace  the r.h.s. of Eq.~A24,
  Property~11, of Ref.~\cite{JM15}.
\end{proof}

\subsection{Extremal R\'enyi entropies under purity constraint}

\begin{lmm}
  Under constraints $\vec\lambda  \ge 0$, $|\vec\lambda|_1 =
  1$,   and  $|\vec\lambda|_2^2   =  P$,   the  extrema   of
  $H_\alpha(\vec\lambda)$  are attained  by a  $\vec\lambda$
  with at most two different non-null eigenvalues, that is
  \begin{align*}
    \vec\lambda =  \left( a_\pm, \dots, a_\pm,  b_\pm, \dots
    b_\pm, 0, \dots 0 \right),
  \end{align*}
  where  $(a_+, b_+)$  and  $(a_-, b_-)$  are  the only  two
  assignments   that  satisfy   the  constraints,   and  are
  explicitly given by
  \begin{align}
    \label{eq:a_pm}
    a_\pm & := \frac{1\pm\sqrt{\frac{n_b}{n_a}\left( P(n_a +
        n_b) -1 \right)}}{n_a + n_b},\\
    \label{eq:b_pm}
    b_\pm & := \frac{1\mp\sqrt{\frac{n_a}{n_b}\left( P(n_a +
        n_b) -1 \right)}}{n_a + n_b},
  \end{align}
  where $n_a$  and $n_b$ denote the  multiplicity of $a_\pm$
  and $b_\pm$, respectively. Explicitly one has
  \begin{align*}
    \min_{\substack{\vec\lambda\ge0\\|\vec\lambda|_1=1\\|\vec\lambda|_2^2
        = P}} H_\alpha\left(\vec\lambda \right) = \min_{n_a,
      n_b,  \pm} \left[  \frac{1}{1-\alpha}  \ln \left(  n_a
      a_\pm^\alpha     +     n_b    b_\pm^\alpha     \right)
      \right],\\ \max_{\substack{\vec\lambda\ge0\\|\vec\lambda|_1=1\\|\vec\lambda|_2^2
        = P}} H_\alpha\left(\vec\lambda \right) = \max_{n_a,
      n_b,  \pm} \left[  \frac{1}{1-\alpha}  \ln \left(  n_a
      a_\pm^\alpha + n_b b_\pm^\alpha \right) \right].
  \end{align*}
\end{lmm}

\begin{proof}
  We consider the following optimization problems
  \begin{align}
    \label{eq:prog1}
    \min_{\substack{\{     \lambda_k      \}_{k=1}^n     \ge
        0\\|\vec\lambda|_1  =   1\\|\vec\lambda|_2^2  =  P}}
    H_\alpha(\vec\lambda),   \qquad    \textrm{and}   \qquad
    \max_{\substack{\{     \lambda_k      \}_{k=1}^n     \ge
        0\\|\vec\lambda|_1  =   1\\|\vec\lambda|_2^2  =  P}}
    H_\alpha(\vec\lambda).
  \end{align}
  We    iteratively     recast    these     equality-    and
  inequality-constrained  programs in  dimension $n$  into a
  set   of   equality-constrained    programs   in   smaller
  dimensions.   Indeed, the  extrema in  Eq.\eqref{eq:prog1}
  are obtained by
  \begin{align*}
    \min_{\substack{\{          \lambda_k         \}_{k=1}^n
        \\|\vec\lambda|_1   =  1\\|\vec\lambda|_2^2   =  P}}
    H_\alpha(\vec\lambda),   \qquad    \textrm{and}   \qquad
    \max_{\substack{\{ \lambda_k \}_{k=1}^n\\|\vec\lambda|_1
        = 1\\|\vec\lambda|_2^2 = P}} H_\alpha(\vec\lambda),
  \end{align*}
  and on the positivity faces, characterized by at least one
  entry equal  to zero  or equivalently by  dimension $n-1$.
  Then, the problem on the positivity faces is
  \begin{align*}
    \min_{\substack{\{    \lambda_k    \}_{k=1}^{n-1}    \ge
        0\\|\vec\lambda|_1  =   1\\|\vec\lambda|_2^2  =  P}}
    H_\alpha(\vec\lambda),   \qquad    \textrm{and}   \qquad
    \max_{\substack{\{    \lambda_k    \}_{k=1}^{n-1}    \ge
        0\\|\vec\lambda|_1  =   1\\|\vec\lambda|_2^2  =  P}}
    H_\alpha(\vec\lambda).
  \end{align*}
  By   iterating,  the   solutions   of   the  programs   in
  Eq.~\eqref{eq:prog1}  are the  solutions  of  this set  of
  programs
  \begin{align}
    \label{eq:prog2}
    \min_{\substack{\{ \lambda_k \}_{k=1}^m\\|\vec\lambda|_1
        = 1\\|\vec\lambda|_2^2  = P}} H_\alpha(\vec\lambda),
    \qquad \textrm{and}  \qquad \max_{\substack{\{ \lambda_k
        \}_{k=1}^m\\|\vec\lambda|_1 = 1\\|\vec\lambda|_2^2 =
        P}} H_\alpha(\vec\lambda),
  \end{align}
  for $1 \le m \le n$.

  We    now     proceed    solving    the     programs    in
  Eq.~\eqref{eq:prog2}.   Notice first  that the  extrema of
  $H_\alpha$ are attained in the  same points as the extrema
  of  $\sum_k  \lambda_k^\alpha$  since $\ln$  is  monotonic
  increasing.  By introducing Lagrange multipliers $\mu$ and
  $\nu$ one has
  \begin{align*}
    F_\alpha := \sum_k \lambda_k^\alpha  - \mu \left( \sum_k
    \lambda_k-1 \right) - \nu  \left( \sum_k \lambda_k^2 - P
    \right),
  \end{align*}
  which for $\alpha = 1$ becomes
  \begin{align*}
    &  F_1  \\ :=  &  -  \sum_k  \lambda_k \ln  \lambda_k  -
    \mu\left(\sum_k  \lambda_k-1\right)   -  \nu\left(\sum_k
    \lambda_k^2 - P\right).
  \end{align*}
  Thus one has
  \begin{align}
    \label{eq:lagrange1}
    \frac{\partial  F_\alpha}{\partial  \lambda_k} =  \alpha
    \lambda_k^{\alpha -1} - \mu -2\nu\lambda_k,
  \end{align}
  which for $\alpha = 1$ becomes
  \begin{align}
    \label{eq:lagrange2}
    \frac{\partial F_1}{\partial  \lambda_k} = -\ln\lambda_k
    - 1 - \mu -2\nu\lambda_k.
  \end{align}

  Since             Eq.~\eqref{eq:lagrange1}             and
  Eq.~\eqref{eq:lagrange2} depend on  $\lambda_k$ only (that
  would not be the case if we had not removed the positivity
  constraint) and  have well-defined concavity,  the optimal
  $\vec\lambda$ has at most  two different non-null entries,
  that we call  $a$ and $b$, and that do  not depend on $k$.
  Then,   the   constraints   $|\vec\lambda|_1  =   1$   and
  $|\vec\lambda|_2^2 =  P$ give the following  conditions on
  $a$ and $b$:
  \begin{align*}
    \begin{cases}
      n_a a + n_b b = 1,\\
      n_a a^2 + n_b b^2 = P.
    \end{cases}
  \end{align*}
  This systems admits two solutions, $(a_+, b_+)$ and $(a_-,
  b_-)$,    as    given     by    Eq.~\eqref{eq:a_pm}    and
  Eq.~\eqref{eq:b_pm}.  Due to the constraint $n_a + n_b \le
  n$, the  number of  such $\vec\lambda$'s is  finite, which
  provides a closed form solution of the optimization.
\end{proof}

\section*{Acknowledgments}

The authors are grateful to  Anna Szymusiak for pointing out
Refs.~\cite{Jon91a,  Jon91b}   and  to  Richard   Jozsa  for
insightful comments based on an  early version of this work.
M.~D.  acknowledges  support from the Ministry  of Education
and the Ministry of  Manpower (Singapore). F.~B acknowledges
support from the Japan Society  for the Promotion of Science
(JSPS) KAKENHI,  Grant No.  17K17796.  This work  was partly
supported by the program  for FRIAS-Nagoya IAR Joint Project
Group.  This  work is  dedicated  to  the memory  of  Graeme
Mitchison.

\begin{IEEEbiographynophoto}{Michele Dall'Arno}
  received  his   PhD  in   theoretical  physics   from  the
  University  of  Pavia,  Italy, in  2011.   After  post-doc
  positions in ICFO, Barcelona, and Nagoya, Japan, he joined
  the Centre  for Quantum Technologies,  National University
  of Singapore, where he is a post-doctoral researcher since
  2014.
\end{IEEEbiographynophoto}

\begin{IEEEbiographynophoto}{Francesco Buscemi}
  received his PhD in theoretical physics from the
  University of Pavia, Italy, in 2006. After post-doc
  positions in Tokyo, Japan, and Cambridge, UK, he joined
  Nagoya University in 2009, where he is a tenured
  associated professor in mathematical informatics since
  2014. In 2018 he was awarded the Birkhoff-von Neumann
  Prize by the International Quantum Structures Association.
\end{IEEEbiographynophoto}

\end{document}